\documentclass{article}

\usepackage{microtype}
\usepackage{graphicx}
\usepackage{subcaption}
\usepackage{booktabs} 
\usepackage{hyperref}

\usepackage[accepted]{icml2026}

\usepackage{amsmath}
\usepackage{amssymb}
\usepackage{mathtools}
\usepackage{amsthm}

\usepackage[capitalize,noabbrev]{cleveref}

\theoremstyle{plain}
\newtheorem{theorem}{Theorem}[section]

\theoremstyle{definition}

\theoremstyle{remark}

\usepackage[textsize=tiny]{todonotes}

\newcommand{\E}{\mathbb{E}}

\usepackage{xcolor}

\icmltitlerunning{SWE-Router}

\begin{document}

\twocolumn[
  \icmltitle{SWE-Router: Routing in Multi-turn Agentic Software Engineering Tasks}

  \icmlsetsymbol{equal}{*}

  \begin{icmlauthorlist}
    \icmlauthor{Seongho Son}{uclcs}
    \icmlauthor{Sangwoong Yoon}{unist}
    \icmlauthor{Jiahua Tang}{psl}
    \icmlauthor{Shuhan Wang}{uclcs}
    \icmlauthor{Lorenz Wolf}{uclcs}
    \icmlauthor{Ilija Bogunovic}{uclcs,basel}
  \end{icmlauthorlist}

  \icmlaffiliation{uclcs}{University College London, United Kingdom}
  \icmlaffiliation{unist}{Ulsan National Institute of Science and Technology, South Korea}
  \icmlaffiliation{psl}{PSL Research University, France}
  \icmlaffiliation{basel}{University of Basel, Switzerland}

  \icmlcorrespondingauthor{Seongho Son}{seong.son.22@ucl.ac.uk}

  \icmlkeywords{Software Engineering, Large Language Models, Routing}

  \vskip 0.3in
]

\printAffiliationsAndNotice{}  

\begin{abstract}
Large language models (LLMs) embedded in multi-turn agentic harnesses are reshaping software engineering (SWE), but routing every task to a frontier model is wasteful when many issues admit cheap fixes. Existing LLM routers operate on the task description alone, which inherits an information-theoretic Bayes-error floor in agentic settings: a similar issue can hide either a localized typo or a multi-module refactor, and the prompt does not separate the two. 
We introduce \textbf{SWE-Router}, a value-based \emph{temporal} approach that lets a cheap model run for a few exploratory turns and reads the resulting partial trajectory before deciding whether to continue cheaply or to escalate to an expensive model. We provide a Bayes-optimality theorem showing that conditioning on the partial trajectory never harms routing and is strictly better whenever exploration is informative.
Across the LLM pairs of weak and strong models spanning the contemporary cost--capability frontier, we show that SWE-Router greatly improves the cost efficiency of SWE tasks, while maintaining the majority of the performances of the stronger model. We additionally release a multi-LLM trajectory dataset which allows reproduction of our trajectory-level routing. \looseness=-1
\end{abstract}

\section{Introduction}

Large language models (LLMs) embedded in multi-turn agentic harnesses~\citep{NEURIPS2024_5a7c9475, wang2024openhands, xia2024agentless, zhang2024autocoderover} are achieving state-of-the-art resolution rates on repository-scale coding benchmarks~\citep{liu2023repobench, jimenez2024swebench, zhuo2024bigcodebench}. However, their per-task inference cost exceeds that of competitive open-weight alternatives by an order of magnitude or more, while only a minority of the existing tasks truly require frontier capability. Routing every instance to a frontier model is therefore wasteful in expectation, which motivates \emph{LLM routing}~\cite{ong2025routellm, chen2024frugalgpt, jitkrittum2025universal, hu2024routerbench}: selecting a model per instance so that expensive inference is reserved for cases in which it is necessary.\looseness=-1

Existing routers operate as classifiers or calibrated probability estimators over the task description $q$~\cite{ong2025routellm, jitkrittum2025universal, ding2024hybrid, aggarwal2023automix, hu2024routerbench}, picking the cheapest model whose cost-adjusted success probability is highest. This is poorly matched to multi-turn agentic SWE: a similar issue description can specify a localized typo or a multi-module refactor, and the distinction is often not identifiable from $q$ alone. 
The information that resolves this ambiguity is \emph{generated by the agent itself}: SWE agents follow ReAct-style loops~\cite{yao2023react} of thoughts, actions, and observations, and a substantial fraction of their turns are spent on bug localization rather than patch synthesis~\cite{xia2024agentless, zhang2024autocoderover}. The intermediate observations --- failed tests, retrieved file contents, stack traces --- supply structural signal that no prompt-only router can access.\looseness=-1

\textbf{SWE-Router}, our value-based \emph{temporal} routing framework, operationalises this insight (\Cref{fig:main-figure}): a cheap weak model $m_1$ runs for a few exploratory turns, then a learned value head reads the resulting partial trajectory and predicts whether $m_1$ will eventually solve the task; if this prediction exceeds a cost-adjusted threshold, $m_1$ continues, otherwise we escalate to a stronger $m_2$. The value head is supervised by binary trajectory rewards and is closely related to execution-free reward models~\cite{shum2026swerm}, but operates on \emph{incomplete} trajectories and feeds a routing rule. 
We also provide a Bayes-optimality result (\Cref{thm:more-information}) showing that partial-trajectory conditioning never harms routing and is strictly better when exploration is informative. 

\begin{figure*}[t]
    \centering    
    \includegraphics[width=\textwidth]{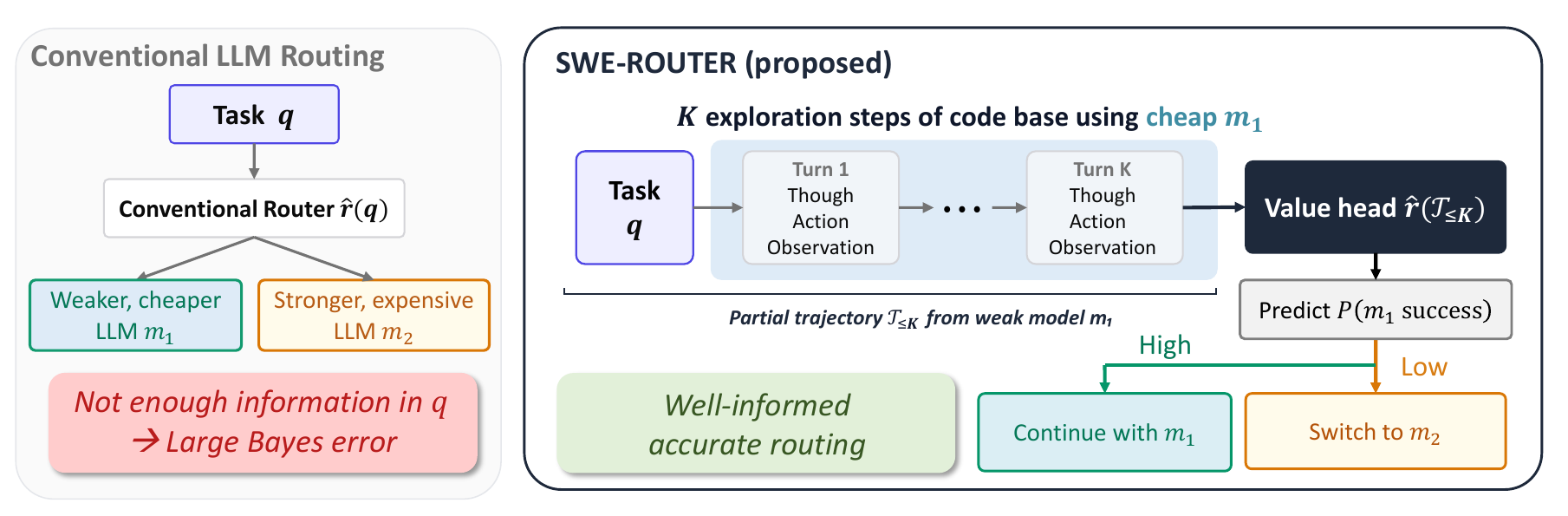}
    \vspace{-10pt}
    \caption{\textbf{SWE-Router overview.} A weak (cheap) model $m_1$ runs for several exploratory turns; a learned value head reads the resulting partial trajectory and predicts whether $m_1$ will eventually solve the task. If the predicted probability exceeds a cost-adjusted threshold, $m_1$ continues; otherwise the strong (expensive) model $m_2$ is invoked from the original prompt $q$. Conditioning on the partial trajectory --- rather than on $q$ alone --- breaks the information-theoretic ceiling that prompt-only routers face (\Cref{thm:more-information}).}
    \label{fig:main-figure}
\end{figure*}

\paragraph{Contributions.} (i) We identify a Bayes-error floor that any prompt-conditioned router inherits in multi-turn agentic SWE, and introduce \textbf{temporal routing} on the partial trajectory; (ii) we instantiate this in \textbf{SWE-Router}, the first routing framework conditioning on the agent's own intermediate observations;
(iii) on SWEBench-Verified, SWE-Router achieves substantial cost reductions at matched resolution relative to strong prompt-only baselines.\looseness=-1

\section{Preliminary}

\subsection{Multi-turn Trajectories in Agentic Tasks}

In \emph{agentic} systems, the LLM autonomously interacts with an environment before producing a final outcome, such as the ReAct framework~\citep{yao2023react} producing a thought $z_t$, an action $a_t$, and an observation $o_t$ each turn~\citep{NEURIPS2024_5a7c9475, yang2024sweagent}.
Given a problem description $q$, a complete SWE-agent trajectory $\mathcal{T}$ takes the form\looseness=-1
\begin{align}
    \mathcal{T}(m)=[q,\quad (z_1, a_1, o_1), \quad \ldots,\quad (z_T,a_T, o_T)],
\end{align}
where $m$ denotes the LLM that generated the trajectory. The final action $a_T$ is \texttt{submit}, whose observation $o_T$ is a git patch. The patch is checked against the repository's unit tests and assigned a verifiable binary reward $r\in \{0,1\}$.\looseness=-1

\subsection{LLM Routing}

An LLM router improves the cost--performance Pareto frontier by invoking a minimum-cost model with sufficient capability. In the binary setting it picks between a \emph{weak} model $m_1$ and a \emph{strong} model $m_2$ given $q$ (``weak/strong'' refers to cost--capability rather than parameter count); model-indexed quantities are subscripted $1, 2$. The router invokes $m_1$ when its cost-adjusted expected success rate exceeds $m_2$'s:\looseness=-1
\begin{align}
    \text{Invoke } \begin{cases}
        m_1 \text{ if } \hat{r}_{1}(q) - \lambda c_1 > \hat{r}_{2}(q) - \lambda c_2 \\
        m_2 \text{ otherwise, } 
    \end{cases}
     \label{eq:router}
\end{align}
where $\hat r_i(q)\in[0,1]$ is the estimated probability $m_i$ resolves $q$, $c_i>0$ is its expected inference cost (USD; on long agentic trajectories cost is model-dominated, so we treat $c_i$ as $q$-independent). Costs are converted into success-probability units by  $\lambda>0$; tuning $\lambda$ traces the cost--quality Pareto curve \citep{ong2025routellm,jitkrittum2025universal}.
\looseness=-1

In agentic settings $q$ often does not pin down task difficulty: early turns reveal feedback (file structure, error patterns, search depth) that the issue text lacks. 
The decision in \Cref{eq:router} therefore frequently inherits large Bayes error, motivating a router that can read the partial trajectory.\looseness=-1

\section{Value-Based Temporal Routing}

We run $m_{1}$ for $K$ turns to collect $\mathcal{T}_{\leq K,1} = [q,(z_1, a_1, o_1), \ldots, (z_K,a_K, o_K)]$ and train a value function $\hat r_1(\mathcal{T}_{\leq K,1})$ as a pretrained LLM with a small classification head, supervised by the binary reward $r$ via cross-entropy. On escalation, $m_2$ \emph{restarts} from $q$: continuation rules need expensive online $m_2$ inference, and conditioning $m_2$ on $m_1$'s reasoning has been seen to bias $m_2$ toward $m_1$'s mistakes. The $K$ weak-model turns are paid up-front and counted in any escalated run's cost. With per-turn weak cost $c<c_1$, the rule is
\begin{align}
    \begin{cases}
        \text{Continue with $m_1$} & \text{if}\quad  y_1 \geq y_2  \\
        \text{Switch to $m_2$ from the $K$-th step} & \text{otherwise},
    \end{cases} \label{eq:value-router}
\end{align}
where $y_1=\hat{r}_1(\mathcal{T}_{\leq K,1}) -\lambda (c_1 - K c)$ and $y_2=\hat{r}_2(q) - \lambda c_2$.
SWE agents typically spend their early turns locating files and isolating failures --- exploration cheaper than patch synthesis, so delegating it to $m_1$ is a useful division of labor. Two simplifications turn \Cref{eq:value-router} into a one-hyperparameter rule: \emph{(i) Threshold reduction:} absorb $\lambda(c_1{-}c_2{-}Kc)$ into a single scalar $\lambda'$ tuned on validation, so ``continue iff $\hat r_1-\hat r_2(q)\ge \lambda'$''; \emph{(ii) One-sided routing:} treat $\hat r_2$ as an unknown trajectory-independent constant, giving ``continue iff $\hat r_1\ge \lambda''$''. We use one-sided routing for all reported experiments. 
\Cref{alg:swe-router} summarizes the deployment-time decision procedure.\looseness=-1

\section{Theoretical Results}

Viewing routing as a one-step Bayesian decision problem, the Bayes-optimal expected utility is monotone in the conditioning $\sigma$-algebra. This is the classical ``information never hurts'' principle for Bayes decision rules~\citep{degroot1962uncertainty,blackwell1953equivalent}, instantiated in our trajectory-conditioning setting. Let $Q$ be the random task description and $S_t \coloneqq [Q, (Z_1, A_1, O_1), \dots, (Z_t, A_t, O_t)]$ the partial trajectory after $t$ weak-model turns; by construction $Q$ is a function of $S_t$, so $S_t$ carries strictly more information than $Q$. The improvement is strict whenever the conditional gap $\E[U_1-U_2\mid S_t]$ is not $\sigma(Q)$-measurable.\looseness=-1

\begin{theorem}[More informative signals improve Bayes-optimal routing]
\label{thm:more-information}
Let $U_1, U_2$ be integrable latent cost-adjusted utilities of routing to $m_1, m_2$, and define $V(Z)\coloneqq\E[\max\{\E[U_1\mid Z], \E[U_2\mid Z]\}]$ for any signal $Z$ (the Bayes-optimal expected utility under $Z$). If $Q$ is determined by $S_t$, then $V(S_t)\ge V(Q)$.
\end{theorem}

\paragraph{Proof sketch.}
The identity $\max\{a,b\} = \tfrac{1}{2}(a+b+|a-b|)$ gives $V(Z) = \tfrac{1}{2}\E[U_1+U_2] + \tfrac{1}{2}\E[|\E[U_1{-}U_2\mid Z]|]$. 
Since $Q$ is a function of $S_t$, the tower rule yields $\E[U_1{-}U_2\mid Q] = \E[\E[U_1{-}U_2\mid S_t]\mid Q]$, and the conditional Jensen inequality applied to $|{\cdot}|$ gives $|\E[\E[U_1{-}U_2\mid S_t]\mid Q]| \le \E[|\E[U_1{-}U_2\mid S_t]|\mid Q]$. Taking expectations yields $V(S_t)\ge V(Q)$. The full proof is in \Cref{app:proof-1}.

\paragraph{Practical reading.} The decomposition above clarifies why temporal routing helps. Prompt-only routers operate at the resolution of $\E[U_1{-}U_2\mid Q]$, which averages out fine-grained difficulty cues only observable after a few exploratory turns. 
The partial trajectory $S_t$ refines this gap: instances on which the two models would make similar predictions get routed cheaply, while hard instances are escalated. The improvement is strict whenever exploration is informative for the cost--benefit comparison; 
it vanishes only when the prompt already determines the optimal choice.\looseness=-1

\begin{algorithm}[t]
\caption{SWE-Router: one-sided routing (deployment).}
\label{alg:swe-router}
\begin{algorithmic}[1]
\REQUIRE problem $q$; weak model $m_1$, strong model $m_2$; value head $\hat r_1$; exploration budget $K$; threshold $\lambda''$ (tuned on validation).
\STATE Initialise $\mathcal{T}\leftarrow [q]$.
\FOR{$t = 1$ to $K$}
    \STATE Run $m_1$ for one step: append $(z_t, a_t, o_t)$ to $\mathcal{T}$.
    \IF{$a_t = \texttt{submit}$} 
    \STATE \textbf{Return} $m_1$'s patch (early termination). \ENDIF
\ENDFOR
\STATE Compute $\hat r_1(\mathcal{T})$ from the value head's last-token logits.
\IF{$\hat r_1(\mathcal{T}) \ge \lambda''$}
    \STATE \textbf{Continue with $m_1$} from turn $K{+}1$ until $a_T=\texttt{submit}$; \textbf{Return} its patch.
\ELSE
    \STATE \textbf{Restart with $m_2$} from $q$; \textbf{Return} $m_2$'s patch.
\ENDIF
\end{algorithmic}
\end{algorithm}

\section{Experiments}
\label{section: experiments}

In this section, we explain how the experiments are implemented and evaluated. We provide further details of experiment in \Cref{appendix: additional experimental settings} and \Cref{app:additional-results}.

\subsection{Evaluation Metrics}
\label{section: evaluation metrics}

\textbf{Threshold-free metrics.} \textbf{AUROC} is the standard ROC-AUC of $\hat r_1(\mathcal{T}_{\le K,1})$ as a binary classifier of $r$ and measures only the value head's discriminative quality. \textbf{Route-AUC} is the normalised area under the cost--vs--resolved-rate curve obtained by sweeping the routing threshold; both axes are normalised so the all-weak point is $(0,0)$ and the all-strong point is $(1,1)$. Hence $0$ corresponds to cost-proportional interpolation between the two LLMs and $1$ matches the strong model's resolved rate at the weak model's cost.\looseness=-1

\subsection{Experiment Setting}
\label{section: experiment setting}

\textbf{Datasets.} We use both SWE-Smith and SWE-Bench for training and test split. We only use SWE-Smith for the validation split. See~\Cref{appendix: additional experimental settings} for further details.

\textbf{LLM agents and pairs.} We generate trajectories from four LLMs in mini-SWE-agent~\citep{yang2024sweagent} and form two (weak, strong) pairs: weak $\in$ \{\texttt{gpt-5-mini}, \texttt{deepseek-v3.2}\}, and \texttt{gemini-3-pro-preview} for the strong model. The weak model has both lower per-instance cost and lower resolved rate in every pair (\Cref{tab:data collection}).

\textbf{Value head and training.} For the value function $\hat{r}_1$, we use \texttt{Qwen2.5-Coder-7B-Instruct} with a 2-class classification head fine-tuned via LoRA~\citep{hu2021loralowrankadaptationlarge}.\looseness=-1

\textbf{Baselines.} Three prompt-only routers in the \Cref{eq:router} family --- logistic regression, $k$-NN, and XGBoost on OpenAI \texttt{text-embedding-3-large} embeddings of $q$ --- and our \textbf{non-temporal router} ($K{=}0$ instance of our framework, conditioning only on $q$) are used as baselines.

\subsection{Results}
\label{section: experiment results}

We evaluate routing on two (weak, strong) LLM pairs on the test splits of SWE-Smith and SWE-Bench datasets. We visualize the curve of routing results in \Cref{fig:op-pareto} and compare the efficiency of SWE-Router against the baselines. The gray bands in the plots indicate the \emph{random-assignment} reference, computed over $400$ Monte-Carlo permutations of the per-instance routing decision. We also use $\times$ and $\bigstar$ to mark the all-weak and all-strong endpoints.\looseness=-1

\begin{figure*}[t]
    \centering
    \includegraphics[width=\textwidth]{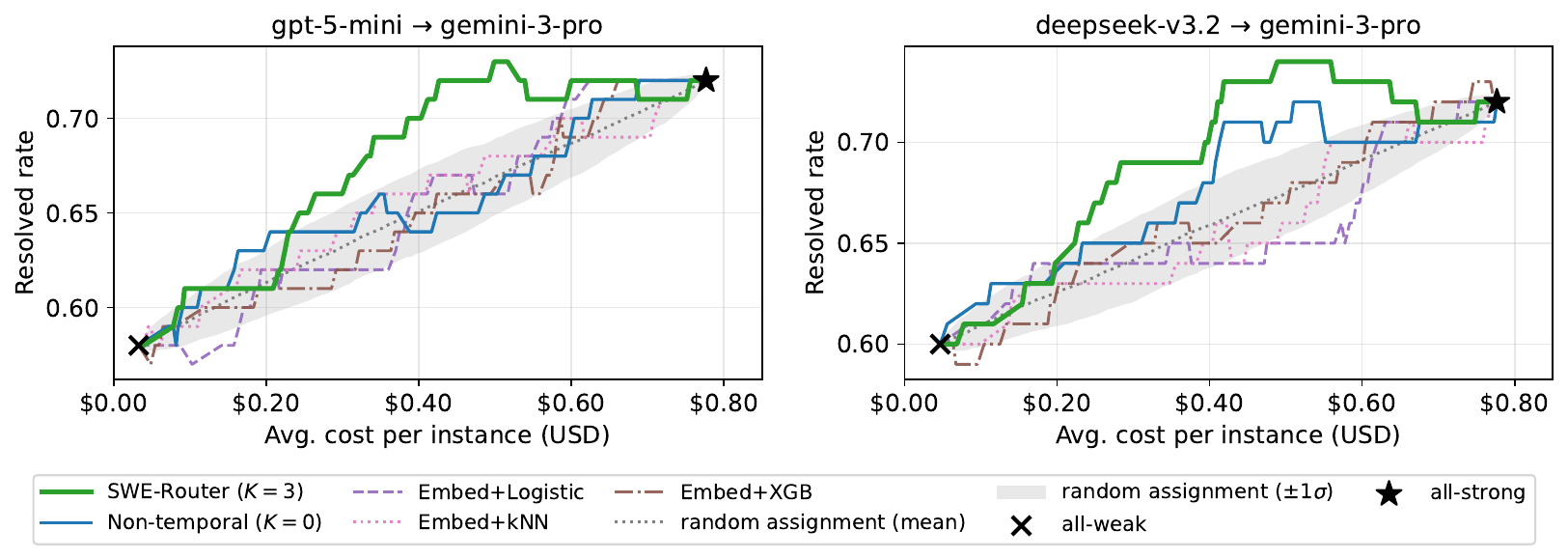}
    \caption{Cost-vs.-resolved \emph{routing curves} on the held-out SWE-Bench Verified for the two (weak, strong) pairs. Each colored curve is the threshold sweep of the corresponding router; the gray band is the \emph{random-assignment} reference, computed as the per-cost-level mean ($\pm1\sigma$) over $400$ Monte-Carlo permutations of the per-instance routing decision; $\times$ and $\bigstar$ mark the all-weak and all-strong endpoints. The curves of SWE-Router ($K{=}3$, green) are above the gray bands, and even exceed the strong-model marker for some thresholds, indicating the synergistic effect of using the weak model for resolving tasks that the strong model could not.}
    \label{fig:op-pareto}
\end{figure*}

\paragraph{Does conditioning on the partial trajectory improve over prompt-only baselines?} The cost--resolved routing curves in \Cref{fig:op-pareto} gives a clear summary of the advantage of using partial trajectory for temporal routing. SWE-Router (green) sits above the random-assignment band in both pairs and traces a Pareto frontier that strictly dominates the embedding-baseline curves over most of the cost range. 
The picture also holds qualitatively for \texttt{deepseek-v3.2} $\to$ \texttt{gemini-3-pro-preview}, where the green curve again passes above the random band at intermediate cost.
We also present the computed ROC-AUC and Route-AUC values of each method in \Cref{tab:full-results-auc-workshop}. In the test split of SWE-Bench Verified, both pairs show significant improvement in Route-AUC over the baselines. When \texttt{deepseek-v3.2} is used as $m_1$, SWE-Router shows Route-AUC of $0.780$, which is a $+15.3$ pp improvement over the non-temporal routing baseline ($K=0$). While the router trained with \texttt{gpt-5-mini}'s trajectories for $m_1$ also exhibits $+12$ pp improvement over the strongest baseline. AUROC evaluations in the test split of SWE-Bench Verified show that the accuracy of value functions are not strongly correlated to actual routing performances, as can be shown from the AUROC of non-temporal router in the \texttt{deepseek-v3.2} $\to$ \texttt{gemini-3-pro-preview} pair.
We further note that on both pairs the SWE-Router curve briefly passes \emph{above} the all-strong $\bigstar$. This synergy is feasible because weak and strong solve non-identical instance subsets, so correctly routing the easy instances to $m_1$ resolves cases on which always-strong fails. The overall shape of the routing curves, also measured through AUC metrics in \Cref{tab:full-results-auc-workshop}, is evidence that the value-head ranking captures a useful complementarity signal for efficient temporal routing.\looseness=-1

\section{Related Work}

Building on agentic frameworks like ReAct~\citep{yao2023react}, recent work equips LLMs to operate on real repositories. SWE-agent~\citep{NEURIPS2024_5a7c9475,yang2024sweagent} introduces an agent--computer interface for browsing, editing, and executing code.
Code benchmarks have moved from self-contained problems~\citep{chen2021evaluating,austin2021program,hendrycks2021measuring} to repository-level evaluation. SWE-bench~\citep{jimenez2024swebench} pioneered repository-scale issue resolution against hidden tests with a curated subset~\citep{swebenchverified}.
FrugalGPT~\citep{chen2024frugalgpt} introduced cost-aware cascades, and RouteLLM~\citep{ong2025routellm} formalized binary routing from preference data for Pareto cost--quality gains. All the existing methods condition on the prompt alone; SWE-Router differs by conditioning on the partial agent trajectory. We provide more related works in \Cref{appendix: Additional Related Work}.\looseness=-1

\section{Conclusion}

We introduced \textbf{SWE-Router}, a value-based routing framework that conditions on a short partial trajectory of a cheap weak model rather than on the task description alone. A simple Bayes-optimality result (\Cref{thm:more-information}) shows that this conditioning never harms routing and is strictly better when exploration is informative. We evaluated our approach using two pairs of popular LLMs, and showed that SWE-Router improves Route-AUC by at least $+12$ pp on the test split of SWE-Bench Verified over the strongest baseline, also presenting the routing curve of methods where SWE-Router achieves the greatest overall performance.\looseness=-1

\clearpage

\section*{Impact Statement}
Our approach provides a clear positive impact of reduced frontier-model inference cost and energy use with minimal tradeoff in the overall performance. However, our proposed approach of utilizing partial trajectories can be misused to bypass safety-focused frontier models for malicious purposes. Considering multiple objectives for routing, such as considering the safety of the given task together during routing, is an important and interesting direction for future work.

\bibliography{refs_icml2026workshop}

@inproceedings{
shum2026swerm,
title={{SWE}-{RM}: Execution-free Feedback for Software Engineering Agents},
author={KaShun SHUM and Binyuan Hui and Jiawei Chen and Lei Zhang and X. W. and Jiaxi Yang and Yuzhen Huang and Junyang Lin and Junxian He},
booktitle={The Fourteenth International Conference on Learning Representations},
year={2026},
url={https://openreview.net/forum?id=H9wMe1G76j}
}

@article{jitkrittum2025universal,
  title={Universal model routing for efficient llm inference},
  author={Jitkrittum, Wittawat and Narasimhan, Harikrishna and Rawat, Ankit Singh and Juneja, Jeevesh and Wang, Congchao and Wang, Zifeng and Go, Alec and Lee, Chen-Yu and Shenoy, Pradeep and Panigrahy, Rina and others},
  journal={arXiv preprint arXiv:2502.08773},
  year={2025}
}

@inproceedings{
zhuang2025embedllm,
title={Embed{LLM}: Learning Compact Representations of Large Language Models},
author={Richard Zhuang and Tianhao Wu and Zhaojin Wen and Andrew Li and Jiantao Jiao and Kannan Ramchandran},
booktitle={The Thirteenth International Conference on Learning Representations},
year={2025},
url={https://openreview.net/forum?id=Fs9EabmQrJ}
}

@inproceedings{
ong2025routellm,
title={Route{LLM}: Learning to Route {LLM}s from Preference Data},
author={Isaac Ong and Amjad Almahairi and Vincent Wu and Wei-Lin Chiang and Tianhao Wu and Joseph E. Gonzalez and M Waleed Kadous and Ion Stoica},
booktitle={The Thirteenth International Conference on Learning Representations},
year={2025},
url={https://openreview.net/forum?id=8sSqNntaMr}
}

@inproceedings{
yao2023react,
title={ReAct: Synergizing Reasoning and Acting in Language Models},
author={Shunyu Yao and Jeffrey Zhao and Dian Yu and Nan Du and Izhak Shafran and Karthik R Narasimhan and Yuan Cao},
booktitle={The Eleventh International Conference on Learning Representations },
year={2023},
url={https://openreview.net/forum?id=WE_vluYUL-X}
}

@inproceedings{
yang2025swesmith,
title={{SWE}-smith: Scaling Data for Software Engineering Agents},
author={John Yang and Kilian Lieret and Carlos E Jimenez and Alexander Wettig and Kabir Khandpur and Yanzhe Zhang and Binyuan Hui and Ofir Press and Ludwig Schmidt and Diyi Yang},
booktitle={The Thirty-ninth Annual Conference on Neural Information Processing Systems Datasets and Benchmarks Track},
year={2025},
url={https://openreview.net/forum?id=63iVrXc8cC}
}

@article{deng2025swe,
  title={Swe-bench pro: Can ai agents solve long-horizon software engineering tasks?},
  author={Deng, Xiang and Da, Jeff and Pan, Edwin and He, Yannis Yiming and Ide, Charles and Garg, Kanak and Lauffer, Niklas and Park, Andrew and Pasari, Nitin and Rane, Chetan and others},
  journal={arXiv preprint arXiv:2509.16941},
  year={2025}
}

@inproceedings{
jimenez2024swebench,
title={{SWE}-bench: Can Language Models Resolve Real-world Github Issues?},
author={Jimenez, Carlos E and Yang, John and Wettig, Alexander and Yao, Shunyu and Pei, Kexin and Press, Ofir and Narasimhan, Karthik},
booktitle={The Twelfth International Conference on Learning Representations},
year={2024},
url={https://openreview.net/forum?id=VTF8yNQM66}
}

@inproceedings{NEURIPS2024_5a7c9475,
 author = {Yang, John and Jimenez, Carlos E. and Wettig, Alexander and Lieret, Kilian and Yao, Shunyu and Narasimhan, Karthik and Press, Ofir},
 booktitle = {Advances in Neural Information Processing Systems},
 doi = {10.52202/079017-1601},
 editor = {A. Globerson and L. Mackey and D. Belgrave and A. Fan and U. Paquet and J. Tomczak and C. Zhang},
 pages = {50528--50652},
 publisher = {Curran Associates, Inc.},
 title = {SWE-agent: Agent-Computer Interfaces Enable Automated Software Engineering},
 url = {https://proceedings.neurips.cc/paper_files/paper/2024/file/5a7c947568c1b1328ccc5230172e1e7c-Paper-Conference.pdf},
 volume = {37},
 year = {2024}
}

@article{chen2021evaluating,
  title={Evaluating Large Language Models Trained on Code},
  author={Chen, Mark and Tworek, Jerry and Jun, Heewoo and Yuan, Qiming and Ponde, Henrique and Kaplan, Jared and Edwards, Harrison and Burda, Yura and Joseph, Nicholas and Brockman, Greg and others},
  journal={arXiv preprint arXiv:2107.03374},
  year={2021}
}

@article{austin2021program,
  title={Program Synthesis with Large Language Models},
  author={Austin, Jacob and Odena, Augustus and Nye, Maxwell and Bosma, Maarten and Michalewski, Henryk and Dohan, David and Jiang, Ellen and Cai, Carrie and Terry, Michael and Le, Quoc and Sutton, Charles},
  journal={arXiv preprint arXiv:2108.07732},
  year={2021}
}

@inproceedings{hendrycks2021measuring,
  title={Measuring Coding Challenge Competence With APPS},
  author={Hendrycks, Dan and Basart, Steven and Kadavath, Saurav and Mazeika, Mantas and Arora, Akul and Guo, Ethan and Burns, Collin and Puranik, Samir and He, Horace and Song, Dawn and Steinhardt, Jacob},
  booktitle={Neural Information Processing Systems},
  year={2021}
}

@misc{swebenchverified, url={https://openai.com/index/introducing-swe-bench-verified/}, journal={Introducing swe-bench verified | openai},author={OpenAI}, year={2024}}

@article{zan2024multiswe,
  title={Multi-SWE-bench: A Multilingual Benchmark for Issue Resolving},
  author={Zan, Daoguang and Huang, Zhirong and Liu, Wei and Chen, Hanwu and Zhang, Linhao and Xin, Shulin and Chen, Lu and Liu, Qi and Zhong, Xiaojian and Li, Aoyan and others},
  journal={arXiv preprint arXiv:2404.02605},
  year={2024}
}

@inproceedings{ding2023crosscodeeval,
  title={CrossCodeEval: A Diverse and Multilingual Benchmark for Cross-File Code Completion},
  author={Ding, Yangruibo and Wang, Zijian and Ahmad, Wasi Uddin and Ding, Hantian and Tan, Ming and Jain, Nihal and Ramanathan, Murali Krishna and Nallapati, Ramesh and Bhatia, Parminder and Roth, Dan and Xiang, Bing},
  booktitle={Neural Information Processing Systems},
  year={2023}
}

@article{liu2023repobench,
  title={RepoBench: Benchmarking Repository-Level Code Auto-Completion Systems},
  author={Liu, Tianyang and Xu, Canwen and McAuley, Julian},
  journal={arXiv preprint arXiv:2306.03091},
  year={2023}
}

@article{zhuo2024bigcodebench,
  title={BigCodeBench: Benchmarking Code Generation with Diverse Function Calls and Complex Instructions},
  author={Zhuo, Terry Yue and Vu, Minh Chien and Chim, Jenny and Hu, Han and Yu, Wenhao and Widyasari, Ratnadira and Yusuf, Imam Nur Bani and Zhan, Haolan and He, Junda and Paul, Indraneil and others},
  journal={arXiv preprint arXiv:2406.15877},
  year={2024}
}

@article{yang2024swebenchmm,
  title={SWE-bench Multimodal: Do AI Systems Generalize to Visual Software Domains?},
  author={Yang, John and Jimenez, Carlos E and Wettig, Alexander and Narasimhan, Karthik and Press, Ofir},
  journal={arXiv preprint arXiv:2410.03859},
  year={2024}
}

@inproceedings{yang2024sweagent,
  title={SWE-agent: Agent-Computer Interfaces Enable Automated Software Engineering},
  author={Yang, John and Jimenez, Carlos E and Wettig, Alexander and Lieret, Kilian and Yao, Shunyu and Narasimhan, Karthik and Press, Ofir},
  booktitle={Neural Information Processing Systems},
  year={2024}
}

@inproceedings{zhang2024autocoderover,
  title={AutoCodeRover: Autonomous Program Improvement},
  author={Zhang, Yuntong and Ruan, Haifeng and Fan, Zhiyu and Roychoudhury, Abhik},
  booktitle={ACM SIGSOFT International Symposium on Software Testing and Analysis},
  year={2024}
}

@article{wang2024openhands,
  title={OpenHands: An Open Platform for AI Software Developers as Generalist Agents},
  author={Wang, Xingyao and Li, Boxuan and Song, Yufan and Xu, Frank F and Tang, Xiangru and Zhuge, Mingchen and Pan, Jiayi and Song, Yueqi and Li, Bowen and Singh, Jaskirat and others},
  journal={arXiv preprint arXiv:2407.16741},
  year={2024}
}

@article{huang2023agentcoder,
  title={AgentCoder: Multi-Agent-based Code Generation with Iterative Testing and Optimisation},
  author={Huang, Dong and Bu, Qingwen and Zhang, Jie M and Luck, Michael and Cui, Heming},
  journal={arXiv preprint arXiv:2312.13010},
  year={2023}
}

@inproceedings{wang2024executable,
  title={Executable Code Actions Elicit Better LLM Agents},
  author={Wang, Xingyao and Chen, Yangyi and Yuan, Lifan and Zhang, Yizhe and Li, Yunzhu and Peng, Hao and Ji, Heng},
  booktitle={International Conference on Machine Learning},
  year={2024}
}

@article{xia2024agentless,
  title={Agentless: Demystifying LLM-based Software Engineering Agents},
  author={Xia, Chunqiu Steven and Deng, Yinlin and Dunn, Soren and Zhang, Lingming},
  journal={arXiv preprint arXiv:2407.01489},
  year={2024}
}

@article{bairi2024codeplan,
  title={CodePlan: Repository-level Coding using LLMs and Planning},
  author={Bairi, Ramakrishna and Sonwane, Atharv and Kanade, Aditya and Vageesh, DC and Iyer, Arun and Parthasarathy, Suresh and Rajamani, Sriram and Ashok, B and Shet, Shashank},
  journal={Proceedings of the ACM on Software Engineering},
  year={2024}
}

@article{xu2024benchmark,
  title={Benchmark Data Contamination of Large Language Models: A Survey},
  author={Xu, Cheng and Guan, Jiuhai and Zhao, Xu and Fu, Chenyi and Xin, Qiushi and Wang, Zihan and Li, Libo and Fu, Jin and Wang, Hao and Liu, Jun},
  journal={arXiv preprint arXiv:2406.04244},
  year={2024}
}

@inproceedings{deng2024investigating,
  title={Investigating Data Contamination in Modern Benchmarks for Large Language Models},
  author={Deng, Chunyuan and Zhao, Yilun and Tang, Xiangru and Gerstein, Mark and Cohan, Arman},
  booktitle={Proceedings of the 2024 Conference of the North American Chapter of the Association for Computational Linguistics: Human Language Technologies (Volume 1: Long Papers)},
  pages={8706--8719},
  year={2024},
  address={Mexico City, Mexico},
  publisher={Association for Computational Linguistics}
}

@article{aleithan2024swe,
  title={SWE-Bench+: Enhanced Coding Benchmark for LLMs},
  author={Aleithan, Reem and others},
  journal={arXiv preprint arXiv:2410.06992},
  year={2024}
}

@article{zhang2025swebench,
  title={SWE-bench Goes Live!},
  author={Zhang, Chaoyun and others},
  journal={arXiv preprint arXiv:2505.23419},
  year={2025}
}

@article{White2024LiveBenchAC,
  title={LiveBench: A Challenging, Contamination-Free LLM Benchmark},
  author={Colin White and Samuel Dooley and ManleyRoberts and Arka Pal and Ben Feuer and Siddhartha Jain and Ravid Shwartz-Ziv and Neel Jain and Khalid Saifullah and Siddartha Naidu and Chinmay Hegde and Yann LeCun and Tom Goldstein and Willie Neiswanger and Micah Goldblum and Abacus.AI and Nyu and Nvidia},
  journal={ArXiv},
  year={2024},
  volume={abs/2406.19314},
  url={https://api.semanticscholar.org/CorpusID:270556394}
}

@article{Da2025AgentRLVRTS,
  title={Agent-RLVR: Training Software Engineering Agents via Guidance and Environment Rewards},
  author={Jeff Da and Clinton J. Wang and Xiang Deng and Yuntao Ma and Nikhil Barhate and Sean M. Hendryx},
  journal={ArXiv},
  year={2025},
  volume={abs/2506.11425},
  url={https://api.semanticscholar.org/CorpusID:279391657}
}

@article{He2025SWEPerfCL,
  title={SWE-Perf: Can Language Models Optimize Code Performance on Real-World Repositories?},
  author={Xinyi He and Qian Liu and Mingzhe Du and Lin Yan and Zhijie Fan and Yiming Huang and Zejian Yuan and Zejun Ma},
  journal={ArXiv},
  year={2025},
  volume={abs/2507.12415},
  url={https://api.semanticscholar.org/CorpusID:280297994}
}

@article{Zhang2024ACE,
  title={A Careful Examination of Large Language Model Performance on Grade School Arithmetic},
  author={Hugh Zhang and Jeff Da and Dean Lee and Vaughn Robinson and Catherine Wu and Will Song and Tiffany Zhao and Pranav Raja and Dylan Slack and Qin Lyu and Sean M. Hendryx and Russell Kaplan and Michele Lunati and Summer Yue},
  journal={ArXiv},
  year={2024},
  volume={abs/2405.00332},
  url={https://api.semanticscholar.org/CorpusID:269484687}
}

@article{feng2024graphrouter,
  title={Graphrouter: A graph-based router for llm selections},
  author={Feng, Tao and Shen, Yanzhen and You, Jiaxuan},
  journal={arXiv preprint arXiv:2410.03834},
  year={2024}
}

@article{ma2025reasoning,
  title={Reasoning models can be effective without thinking},
  author={Ma, Wenjie and He, Jingxuan and Snell, Charlie and Griggs, Tyler and Min, Sewon and Zaharia, Matei},
  journal={arXiv preprint arXiv:2504.09858},
  year={2025}
}

@article{huang2025routereval,
  title={Routereval: A comprehensive benchmark for routing llms to explore model-level scaling up in llms},
  author={Huang, Zhongzhan and Ling, Guoming and Lin, Yupei and Chen, Yandong and Zhong, Shanshan and Wu, Hefeng and Lin, Liang},
  journal={arXiv preprint arXiv:2503.10657},
  year={2025}
}

@article{
    chen2024frugalgpt,
    title={Frugal{GPT}: How to Use Large Language Models While Reducing Cost and Improving Performance},
    author={Lingjiao Chen and Matei Zaharia and James Zou},
    journal={Transactions on Machine Learning Research},
    issn={2835-8856},
    year={2024},
    url={https://openreview.net/forum?id=cSimKw5p6R},
    note={}
}

@inproceedings{vsakota2024fly,
  title={Fly-swat or cannon? cost-effective language model choice via meta-modeling},
  author={{\v{S}}akota, Marija and Peyrard, Maxime and West, Robert},
  booktitle={Proceedings of the 17th ACM International Conference on Web Search and Data Mining},
  pages={606--615},
  year={2024}
}

@article{chen2024routerdc,
  title={Routerdc: Query-based router by dual contrastive learning for assembling large language models},
  author={Chen, Shuhao and Jiang, Weisen and Lin, Baijiong and Kwok, James and Zhang, Yu},
  journal={Advances in Neural Information Processing Systems},
  volume={37},
  pages={66305--66328},
  year={2024}
}

@article{hari2023tryage,
  title={Tryage: Real-time, intelligent routing of user prompts to large language models},
  author={Hari, Surya Narayanan and Thomson, Matt},
  journal={arXiv preprint arXiv:2308.11601},
  year={2023}
}

@article{shirkavand2025cost,
  title={Cost-Aware Contrastive Routing for LLMs},
  author={Shirkavand, Reza and Gao, Shangqian and Yu, Peiran and Huang, Heng},
  journal={arXiv preprint arXiv:2508.12491},
  year={2025}
}

@article{ding2024hybrid,
  title={Hybrid llm: Cost-efficient and quality-aware query routing},
  author={Ding, Dujian and Mallick, Ankur and Wang, Chi and Sim, Robert and Mukherjee, Subhabrata and Ruhle, Victor and Lakshmanan, Laks VS and Awadallah, Ahmed Hassan},
  journal={arXiv preprint arXiv:2404.14618},
  year={2024}
}

@article{aggarwal2023automix,
  title={Automix: Automatically mixing language models},
  author={Aggarwal, Pranjal and Madaan, Aman and Anand, Ankit and Potharaju, Srividya Pranavi and Mishra, Swaroop and Zhou, Pei and Gupta, Aditya and Rajagopal, Dheeraj and Kappaganthu, Karthik and Yang, Yiming and others},
  journal={arXiv preprint arXiv:2310.12963},
  year={2023}
}

@article{jiang2023llm,
  title={Llm-blender: Ensembling large language models with pairwise ranking and generative fusion},
  author={Jiang, Dongfu and Ren, Xiang and Lin, Bill Yuchen},
  journal={arXiv preprint arXiv:2306.02561},
  year={2023}
}

@inproceedings{
    yue2024large,
    title={Large Language Model Cascades with Mixture of Thought Representations for Cost-Efficient Reasoning},
    author={Murong Yue and Jie Zhao and Min Zhang and Liang Du and Ziyu Yao},
    booktitle={The Twelfth International Conference on Learning Representations},
    year={2024},
    url={https://openreview.net/forum?id=6okaSfANzh}
}

@article{somerstep2025carrot,
  title={Carrot: A cost aware rate optimal router},
  author={Somerstep, Seamus and Polo, Felipe Maia and de Oliveira, Allysson Flavio Melo and Mangal, Prattyush and Silva, M{\'\i}rian and Bhardwaj, Onkar and Yurochkin, Mikhail and Maity, Subha},
  journal={arXiv preprint arXiv:2502.03261},
  year={2025}
}

@article{feng2025ipr,
  title={IPR: Intelligent Prompt Routing with User-Controlled Quality-Cost Trade-offs},
  author={Feng, Aosong and Xu, Zhichao and Wu, Xian and Zhou, Kang and Guan, Sheng and Chen, Yueyan and Kulkarni, Ninad and Zhou, Yun and Srinivasan, Balasubramaniam and Ding, Haibo and others},
  journal={arXiv preprint arXiv:2509.06274},
  year={2025}
}

@article{mei2025omnirouter,
  title={OmniRouter: Budget and Performance Controllable Multi-LLM Routing},
  author={Mei, Kai and Xu, Wujiang and Lin, Shuhang and Zhang, Yongfeng},
  journal={arXiv preprint arXiv:2502.20576},
  year={2025}
}

@article{lu2025routerarena,
  title={RouterArena: An Open Platform for Comprehensive Comparison of LLM Routers},
  author={Lu, Yifan and Liu, Rixin and Yuan, Jiayi and Cui, Xingqi and Zhang, Shenrun and Liu, Hongyi and Xing, Jiarong},
  journal={arXiv preprint arXiv:2510.00202},
  year={2025}
}

@article{song2025irt,
  title={IRT-Router: Effective and Interpretable Multi-LLM Routing via Item Response Theory},
  author={Song, Wei and Huang, Zhenya and Cheng, Cheng and Gao, Weibo and Xu, Bihan and Zhao, GuanHao and Wang, Fei and Wu, Runze},
  journal={arXiv preprint arXiv:2506.01048},
  year={2025}
}

@article{pan2025route,
  title={Route to Reason: Adaptive Routing for LLM and Reasoning Strategy Selection},
  author={Pan, Zhihong and Zhang, Kai and Zhao, Yuze and Han, Yupeng},
  journal={arXiv preprint arXiv:2505.19435},
  year={2025}
}

@article{shao2025route,
  title={Route-and-Reason: Scaling Large Language Model Reasoning with Reinforced Model Router},
  author={Shao, Chenyang and Liu, Xinyang and Lin, Yutang and Xu, Fengli and Li, Yong},
  journal={arXiv preprint arXiv:2506.05901},
  year={2025}
}

@misc{Guo2026ThinkWN,
      title={Think When Needed: Model-Aware Reasoning Routing for LLM-based Ranking}, 
      author={Huizhong Guo and Tianjun Wei and Dongxia Wang and Yingpeng Du and Ziyan Wang and Jie Zhang and Zhu Sun},
      year={2026},
      eprint={2601.18146},
      archivePrefix={arXiv},
      primaryClass={cs.IR},
      url={https://arxiv.org/abs/2601.18146}, 
}

@article{wang2025reason,
  title={When to Reason: Semantic Router for vLLM},
  author={Wang, Chen and Liu, Xunzhuo and Liu, Yuhan and Zhu, Yue and Mo, Xiangxi and Jiang, Junchen and Chen, Huamin},
  journal={arXiv preprint arXiv:2510.08731},
  year={2025}
}

@article{hu2024routerbench,
  title={Routerbench: A benchmark for multi-llm routing system},
  author={Hu, Qitian Jason and Bieker, Jacob and Li, Xiuyu and Jiang, Nan and Keigwin, Benjamin and Ranganath, Gaurav and Keutzer, Kurt and Upadhyay, Shriyash Kaustubh},
  journal={arXiv preprint arXiv:2403.12031},
  year={2024}
}

@inproceedings{zhang2025router,
  title={Router-r1: Teaching llms multi-round routing and aggregation via reinforcement learning},
  author={Zhang, Haozhen and Feng, Tao and You, Jiaxuan},
  booktitle={The Thirty-ninth Annual Conference on Neural Information Processing Systems},
  year={2025}
}

@misc{hu2021loralowrankadaptationlarge,
      title={LoRA: Low-Rank Adaptation of Large Language Models}, 
      author={Edward J. Hu and Yelong Shen and Phillip Wallis and Zeyuan Allen-Zhu and Yuanzhi Li and Shean Wang and Lu Wang and Weizhu Chen},
      year={2021},
      eprint={2106.09685},
      archivePrefix={arXiv},
      primaryClass={cs.CL},
      url={https://arxiv.org/abs/2106.09685}, 
}

@article{blackwell1953equivalent,
  title={Equivalent comparisons of experiments},
  author={Blackwell, David},
  journal={The Annals of Mathematical Statistics},
  volume={24},
  number={2},
  pages={265--272},
  year={1953}
}

@article{degroot1962uncertainty,
  title={Uncertainty, information, and sequential experiments},
  author={DeGroot, Morris H.},
  journal={The Annals of Mathematical Statistics},
  volume={33},
  number={2},
  pages={404--419},
  year={1962}
}
\bibliographystyle{icml2026}

\newpage
\appendix
\onecolumn

\appendix

\section{Additional Experimental Settings}
\label{appendix: additional experimental settings}

\subsection{Data Collection}
\label{section:data collection}
Training value functions requires labeled trajectory data from weaker LLMs; since SWE-smith only ships \texttt{claude-3.7-sonnet} trajectories~\citep{yang2025swesmith}, we generate our own and release them. \Cref{tab:data collection} summarizes the collected splits.

\begin{table*}[t]
    \centering
    \footnotesize
    \setlength{\tabcolsep}{4pt}
    \begin{tabular}{l|cccc}
        \toprule
        Dataset / Split & \texttt{gpt-5-mini} & \texttt{deepseek-v3.2} & \texttt{gemini-3-pro} \\
        \midrule
        SWE-Smith Train        & 1{,}692 / 33.5 / \$70.1 & 1{,}755 / 31.9 / \$85.7 & ---                   \\
        SWE-Smith Val          & 203 / 31.0 / \$7.7      & 210 / 27.5 / \$8.9      & 210 / 32.3 / \$121.3       \\
        SWE-Smith Test         & 339 / 48.7 / \$11.0     & 346 / 42.7 / \$13.3     & 173 / 55.4 / \$95.6  \\
        SWE-Bench Verified     & 500 / 55.6 / \$16.6     & 500 / 62.4 / \$22.6     & 500 / 70.4 / \$391.3 \\
        \midrule
        \emph{Total cost (\$)} & 105.4                   & 130.5                   & 608.2                \\
        \bottomrule
    \end{tabular}
    \caption{Data collection (mix-1 constitution; \Cref{appendix: experiment designs}). Each cell shows \emph{instances} / \emph{resolved \%} / \emph{cost (USD)}, from running each LLM in mini-SWE-agent with a 75-step budget. ``---'' marks splits not collected for that LLM. \texttt{gemini-3-pro-preview} does not have training split because SWE-Router only needs it from weak models used for $m_1$.}
    \label{tab:data collection}
\end{table*}

\textbf{SWE-Smith instances.} We sub-sample SWE-smith into three repository-disjoint partitions: \emph{train} ($\sim$1.7k trajectories per weak LLM), \emph{val} ($\sim$210, used for threshold tuning), and \emph{test} ($\sim$170--346, drawn from yet another repository set). Resolved/unresolved labels come from running the SWE-smith unit tests against the agent's submitted patch.

\textbf{SWE-Bench instances.} For cross-distribution evaluation we collect trajectories on SWE-bench Verified ($500$ issues, all four LLMs). Under mix-1 (\Cref{appendix: experiment designs}), $4/5$ of these instances enter training, so the genuinely held-out slice consists of 100 instances.

\subsection{Experiment Designs}
\label{appendix: experiment designs}

\textbf{Mix-1 split constitution.} \emph{Train} merges SWE-Smith train trajectories with $4/5$ of a 5-fold-CV partition of SWE-Bench Verified. We construct a validation split using the tasks from SWE-Smith, while having two separate test split from each of SWE-Smith and SWE-Bench Verified.

\textbf{Value head and training.} The LoRA adapter for $\hat{r}_1$  uses $r{=}32$, $\alpha{=}64$, dropout $0.05$ on attention/MLP projections. For training, we use $5$ epochs, lr $5\times 10^{-5}$ cosine with $50$-step warmup, effective batch $16$, and context length limit of $8192$ tokens. We use a \emph{packed} procedure: one row per trajectory and the head is applied at every user-turn boundary $K\in\{0,\dots,K_{\max}\}$ via a single shared forward pass. Each $q$ is replaced by a uniformly-sampled paraphrase from a 3-way LLM rephrasing during training and validation, and we set $K_{\max}=4$.

\textbf{Baselines.} Three prompt-only routers in the \Cref{eq:router} family --- logistic regression, $k$-NN, and XGBoost on OpenAI \texttt{text-embedding-3-large} embeddings of $q$ --- and our \textbf{non-temporal router} ($K{=}0$ instance of our framework, conditioning only on $q$).

\section{Additional Experiment Results}
\label{app:additional-results}

\Cref{tab:full-results-auc-workshop} reports the full set of routing experiment results. Rows are grouped by the (weak, strong) model pair and list, in order: three prompt-only embedding baselines (\emph{Embed+LR}, \emph{Embed+kNN}, \emph{Embed+XGB}); our non-temporal variant ($K{=}0$); SWE-Router with fixed budgets $K\in\{1,2,3,4\}$. Metric semantics and the splits used in the experiment are explained in \Cref{section: experiments} and \Cref{appendix: additional experimental settings}.

\begin{table}[t!]
\centering
\caption{Threshold-free routing metrics (AUROC and Route-AUC) on three splits (SWE-Smith val, SWE-Smith test, SWE-Bench Verified test). ``---'' indicates an experiment not run or a metric not produced by the corresponding evaluation script. SWE-Router significantly outperforms all the baselines in SWE-Bench Verified test split in Route-AUC, while in SWE-Smith test split, the router trained with \texttt{gpt-5-mini} as $m_1$ does not show higher Route-AUC than the baselines. As the Route-AUC between the validation split and test split of SWE-Smith is not consistently correlated in the case of router with \texttt{deepseek-v3.2} as $m_1$ as well, we attribute this phenomenon to a distribution shift between the splits while separating them based on the code repositories the tasks are derived from.}
\label{tab:full-results-auc-workshop}
\footnotesize
\setlength{\tabcolsep}{4pt}
\begin{tabular}{l|ccc|ccc}
\toprule
& \multicolumn{3}{c|}{AUROC} & \multicolumn{3}{c}{Route-AUC}\\
Method & val & SS-test & SB-V & val & SS-test & SB-V \\
\midrule
\multicolumn{7}{l}{\textbf{$m_1$: \texttt{gpt-5-mini} \quad $m_2$: \texttt{gemini-3-pro-preview}}} \\
Embed+LR & 0.635 & 0.692 & 0.527 & 1.140 & 0.591 & 0.565 \\
Embed+kNN & 0.629 & 0.664 & 0.563 & -0.301 & 0.604 & 0.530 \\
Embed+XGB & 0.585 & 0.589 & 0.567 & 0.782 & 0.521 & 0.517 \\
Non-temporal Router (ours, $K{=}0$) & 0.599 & 0.702 & 0.613 & 1.644 & 0.626 & 0.549 \\
SWE-Router fixed ($K{=}1$) & 0.622 & 0.679 & 0.597 & 2.358 & 0.493 & 0.702 \\
SWE-Router fixed ($K{=}2$) & 0.634 & 0.678 & 0.594 & 1.928 & 0.482 & 0.703 \\
SWE-Router fixed ($K{=}3$) & 0.631 & 0.684 & 0.586 & 1.982 & 0.547 & 0.694 \\
SWE-Router fixed ($K{=}4$) & 0.636 & 0.685 & 0.617 & 2.243 & 0.546 & 0.709 \\
\midrule
\multicolumn{7}{l}{\textbf{$m_1$: \texttt{deepseek-v3.2} \quad $m_2$: \texttt{gemini-3-pro-preview}}} \\
Embed+LR & 0.628 & 0.660 & 0.477 & 0.455 & 0.508 & 0.476 \\
Embed+kNN & 0.600 & 0.590 & 0.489 & 0.385 & 0.444 & 0.371 \\
Embed+XGB & 0.591 & 0.618 & 0.526 & 0.430 & 0.555 & 0.465 \\
Non-temporal Router (ours, $K{=}0$) & 0.628 & 0.638 & 0.620 & 0.663 & 0.555 & 0.627 \\
SWE-Router fixed ($K{=}1$) & 0.658 & 0.649 & 0.597 & 0.646 & 0.571 & 0.768 \\
SWE-Router fixed ($K{=}2$) & 0.670 & 0.654 & 0.605 & 0.720 & 0.567 & 0.780 \\
SWE-Router fixed ($K{=}3$) & 0.666 & 0.655 & 0.603 & 0.738 & 0.547 & 0.750 \\
SWE-Router fixed ($K{=}4$) & 0.651 & 0.669 & 0.596 & 0.668 & 0.569 & 0.718 \\
\bottomrule
\end{tabular}
\end{table}

\clearpage
\section{Proof of \Cref{thm:more-information}}\label{app:proof-1}

\begin{proof}
Let $\Delta \coloneqq U_1-U_2$. Using the identity
\[
\max\{a,b\} = \tfrac12(a+b+|a-b|),
\]
we obtain, for any signal $Z$,
\[
V(Z) \;=\; \tfrac12 \E[U_1+U_2] \;+\; \tfrac12\,\E\!\left[\,\big|\E[\Delta \mid Z]\big|\,\right].
\]
The first term does not depend on $Z$, so
\[
V(S_t) - V(Q) \;=\; \tfrac12 \left(\E\!\left[|\E[\Delta \mid S_t]|\right] \;-\; \E\!\left[|\E[\Delta \mid Q]|\right]\right).
\]
Since $Q$ is a function of $S_t$, the tower property of conditional expectation gives
\[
\E[\Delta \mid Q] \;=\; \E\!\left[\,\E[\Delta \mid S_t] \,\big|\, Q\,\right].
\]
Applying Jensen's inequality with the convex function $\varphi(x)=|x|$ conditionally on $Q$,
\[
\big|\E[\Delta \mid S_t] \mid Q\,\big| \;\le\; \E\!\left[\,|\E[\Delta \mid S_t]| \,\big|\, Q\,\right],
\]
and taking outer expectations on both sides yields
\[
\E\!\left[|\E[\Delta\mid Q]|\right] \;\le\; \E\!\left[|\E[\Delta\mid S_t]|\right].
\]
Hence $V(S_t)\ge V(Q)$, and equality holds iff $\E[\Delta\mid S_t]$ is $\sigma(Q)$-measurable, i.e., the partial trajectory carries no information about $\Delta$ beyond the prompt.
\end{proof}

\section{Additional Related Work}
\label{appendix: Additional Related Work}

\textbf{LLMs for software engineering.}
OpenHands~\citep{wang2024openhands} provides an open platform for generalist software agents. Other approaches localize and patch bugs via program structure~\citep{zhang2024autocoderover}, simple two-stage pipelines~\citep{xia2024agentless}, planning~\citep{bairi2024codeplan}, multi-agent testing~\citep{huang2023agentcoder}, or executable code actions~\citep{wang2024executable}. For training, SWE-smith~\citep{yang2025swesmith} synthesizes bug-and-fix data by injecting defects into repositories, and Agent-RLVR~\citep{Da2025AgentRLVRTS} uses execution-based RL rewards. SWE-RM~\citep{shum2026swerm} instead provides execution-free reward signals for patch scoring; our value head is related but operates on \emph{partial} trajectories.

\textbf{Software engineering benchmarks.}
Repository-level evaluations often involve cross-file completion~\citep{ding2023crosscodeeval,liu2023repobench} and complex tool use~\citep{zhuo2024bigcodebench}. SWE-bench~\citep{jimenez2024swebench} now has extensions to multimodal~\citep{yang2024swebenchmm}, multilingual~\citep{zan2024multiswe}, higher-rigor~\citep{aleithan2024swe}, live~\citep{zhang2025swebench}, and performance~\citep{He2025SWEPerfCL} settings. Since these tasks are scraped from public code, contamination is a recurring concern~\citep{xu2024benchmark,deng2024investigating}, motivating contamination-resistant alternatives~\citep{White2024LiveBenchAC,Zhang2024ACE}; SWE-bench Pro~\citep{deng2025swe} addresses this with private repositories.

\textbf{LLM routing.}
Existing approaches to routing between LLMs combine models of differing capability via hybrid routing~\citep{ding2024hybrid}, mixing~\citep{aggarwal2023automix}, meta-modeling~\citep{vsakota2024fly}, or output ensembling~\citep{jiang2023llm}. Subsequent routers use contrastive embeddings~\citep{chen2024routerdc}, real-time prompt classification~\citep{hari2023tryage}, graph formulations~\citep{feng2024graphrouter}, and explicit cost or budget controls~\citep{shirkavand2025cost,somerstep2025carrot,mei2025omnirouter,feng2025ipr}; others target reasoning vs.\ non-reasoning routing~\citep{ma2025reasoning,pan2025route,shao2025route,wang2025reason,Guo2026ThinkWN} and thought-level cascades~\citep{yue2024large}. To handle unseen models, EmbedLLM~\citep{zhuang2025embedllm} and universal model routing~\citep{jitkrittum2025universal} embed both prompts and LLMs so routers generalize beyond training. RL-based routers~\citep{zhang2025router,song2025irt} learn multi-round or interpretable selection policies, and benchmarks~\citep{hu2024routerbench,huang2025routereval,lu2025routerarena} consolidate the field.

\end{document}